\documentclass[sigconf,nonacm]{aamas}

\usepackage{balance}
\usepackage{amsmath}
\usepackage{tabularx}
\usepackage{array}
\usepackage{bbm}
\usepackage{tikz}

\usetikzlibrary{positioning,arrows.meta,calc}

\usepackage{pgfplots}
\pgfplotsset{compat=1.18}
\usepgfplotslibrary{groupplots}

\setcopyright{none}
\title[Evidence-Grounded Diagnosis of Collective Mechanisms]{
Where Do Multi-Agent Systems Fail?
Evidence-Grounded Diagnosis of Collective Mechanisms
}

\author{Zhengye Han}

\affiliation{%
  \institution{New York University}
  \department{Department of Electrical and Computer Engineering}
  \city{Brooklyn}
  \state{New York}
  \country{USA}
}

\email{zh3286@nyu.edu}

\begin{abstract}
When a multi-agent system answers correctly, it is tempting to conclude
that its agents shared, checked, and used information as intended.  Yet
a system can break one of its collective mechanisms, the rules that
govern how agents route, admit, store, and act on shared information,
and still return the right answer, while a wrong answer rarely reveals
which mechanism failed.  We ask what evidence from an execution is
sufficient to conclude that a particular mechanism was violated.  Our
answer is a diagnostic contract, which separates what counts as a
violation from which execution records can establish one, and concludes
that a violation is supported, ruled out, or unknown; when records are removed and the rest are left unchanged, a resolved conclusion can become unknown but cannot reverse.  We test
contracts for four mechanisms by replaying executions from the step
where a mechanism acts, once unchanged, once with the mechanism broken,
and once with it restored.  A broken mechanism could leave the answer correct.  An LLM diagnoser detected many more violations from internal
records than from public outputs, yet with identical records a generic
prompt often claimed certainty the records did not support, which
prompts stating the contracts largely avoided.  The contracts also
applied, in narrow form, to mechanisms in independently developed
systems, but a diagnostic behavior that was nearly perfect on our
benchmark degraded on an independently developed workflow.  A correct
outcome is therefore no substitute for records of how collective
mechanisms operated, and agreement on one benchmark does not show that
a diagnoser transfers to another system.

\end{abstract}

\ccsdesc[500]{Computing methodologies~Multi-agent systems}

\keywords{multi-agent systems, collective mechanisms, failure diagnosis, diagnostic contracts, execution evidence}

\newcommand{\loci}{\mathcal{L}}

\begin{document}
\maketitle

\section{Introduction}
\label{sec:intro}

\begin{figure*}[t]
\centering
\includegraphics[width=0.99\textwidth]{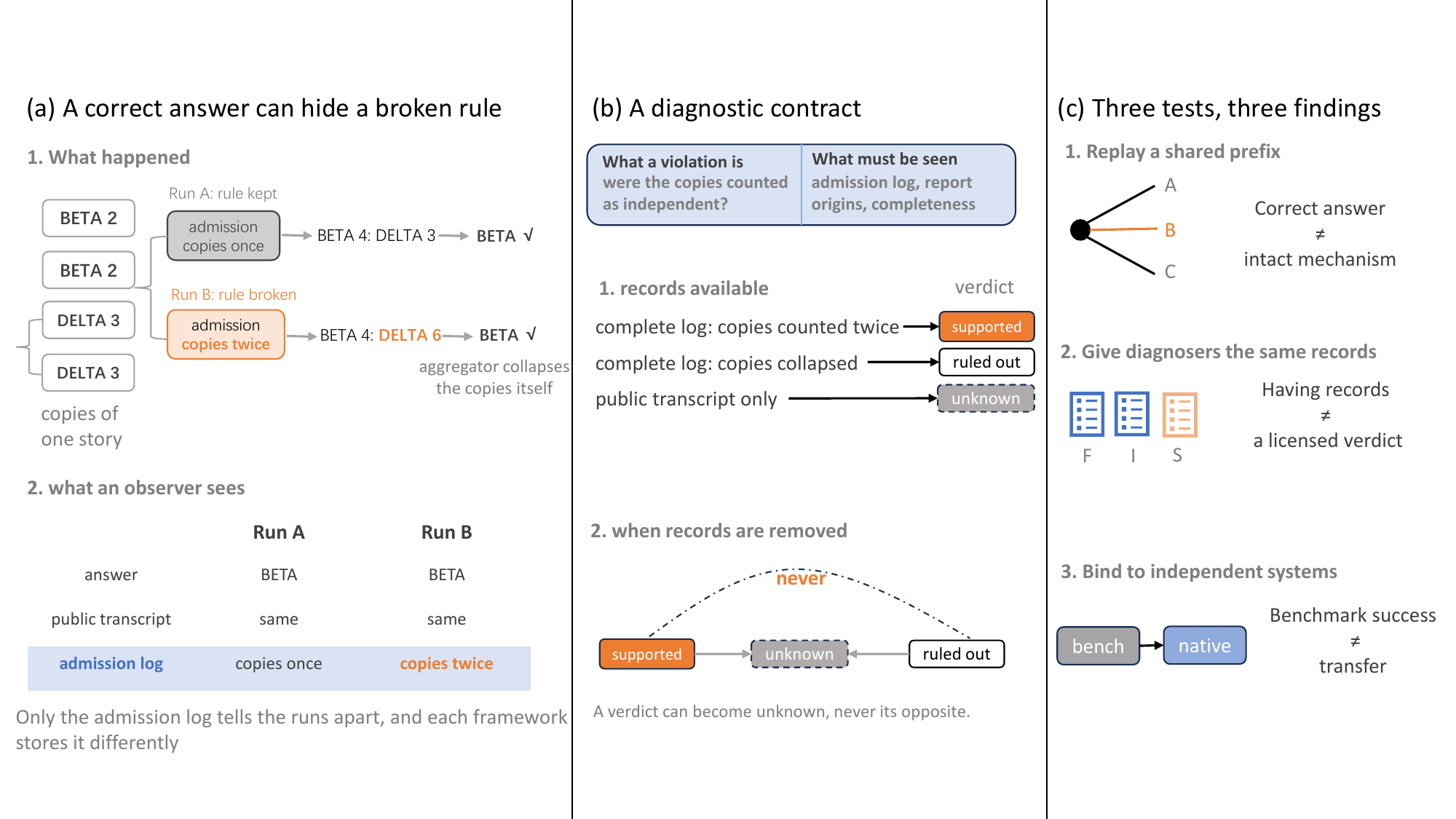}
\caption{\textbf{Where a multi-agent system fails, and how we diagnose it.}
(a)~Two runs continue from the same reports.  Run~B breaks the admission
rule by counting two copies of one story as independent, yet both runs
answer \textsc{beta} with the same public transcript; only the admission
log, which each framework stores differently, tells them apart.
(b)~A diagnostic contract separates what a violation \emph{is} from what
must be \emph{seen} and returns \emph{supported}, \emph{ruled out}, or
\emph{unknown}; removing records can turn a verdict into \emph{unknown}
but never into its opposite.  (c)~Three tests, one finding each.}
\Description{Three panels.  Panel a: four reports (beta 2, beta 2, and
two copies of one story with delta 3 each) feed two runs.  In Run A the
admission step counts the copies once, the tally is beta 4 to delta 3,
and the answer is beta.  In Run B the admission step counts them twice,
the tally is beta 4 to delta 6, and the answer is still beta because the
aggregator collapses the copies itself.  A table below shows that the
answer and the public transcript are the same in both runs and only the
admission log differs.  Panel b: a contract box states what a violation
is and what must be seen.  A complete log showing the copies counted
twice yields supported, a complete log showing them collapsed yields
ruled out, and the public transcript alone yields unknown.  When records
are removed, supported or ruled out can become unknown but never each
other.  Panel c: three tests and their findings.  Replaying a shared
prefix: a correct answer is not an intact mechanism.  Giving diagnosers
the same records: having records is not a licensed verdict.  Binding to
independently developed systems: benchmark success is not transfer.}
\label{fig:protocol}
\end{figure*}


In a multi-agent system, agents do not only reason.  They pass
information to one another, decide which reports to trust, keep shared
state, and turn decisions into actions.  We call the rules that govern
these steps \emph{collective mechanisms}.  When such a system
misbehaves, repairing it requires knowing which mechanism failed:
whether a failure is better repaired by changing the rules or by
strengthening the agents depends on where it
arose~\cite{han2026institutions}.  A wrong diagnosis therefore leads to
the wrong repair, and a violation that no one detects stays in the
system.  The direct evidence about these mechanisms is the
\emph{execution record}: the messages, inputs, state updates, and
action receipts that a run leaves behind.

Recent work has made failed multi-agent executions much easier to
analyze.  Failure taxonomies classify what goes wrong in annotated
traces~\cite{cemri2025mast}; attribution methods identify the agent and
step responsible for a failure~\cite{zhang2025whowhen,chen2026traceelephant};
constraint-based methods localize critical failure
steps~\cite{barke2026agentrx}; and fault-injection benchmarks measure how
known faults are detected and
recovered~\cite{jia2026masfire,zhang2026agentchaos,chen2026orchestra}.
These approaches typically start from an unsuccessful outcome or a known
fault and return a label for what went wrong.

Diagnosing collective mechanisms, however, raises three difficulties
that outcome-driven analysis does not address.  The first is
\emph{concealment}: a violated mechanism need not produce a wrong answer.
Consider a system that must choose between two verdicts, \textsc{beta}
and \textsc{delta}.  Source agents forward reports to an aggregator, and
an admission step decides which reports the aggregator may use.  Two
independent sources support \textsc{beta} with weight two each; two
reports support \textsc{delta} with weight three each, but both are
copies of one wire story, and the admission rule counts such copies
once.  In one execution of our benchmark, the admission step counted
both copies, favoring \textsc{delta} by six to four, yet the system
answered \textsc{beta}: the aggregator's recorded reasoning shows that
its own instructions restated the rule and that it collapsed the copies
itself.  The run was scored as a success, so an analysis triggered by
failure would never examine it.  The second difficulty is
\emph{ambiguity}.  The public transcript of this run is equally
consistent with an execution in which admission worked, a label such as
``admission failure'' does not say which record would justify it, and a
missing record is not evidence that an event did not occur.  The third
is \emph{heterogeneity}.  Frameworks record different things, so a
diagnosis read off one framework's traces may not mean the same on
another, and a diagnostic procedure that works on one set of systems
may not transfer to the next.  These difficulties lead to one question:

\begin{quote}
\emph{Given an execution record, what evidence is sufficient to conclude
that a particular collective mechanism was violated?}
\end{quote}

We answer it with a \emph{diagnostic contract}, which separates what a
violation \emph{is} from what an observer must \emph{see} to establish
it, and each of its parts addresses one difficulty.  Because the
contract defines a violation over the operation of a mechanism rather
than over the task outcome, it applies to successful runs as well as
failed ones; for admission, a violation occurs when evidence is accepted
or weighted contrary to the admission rule.  Because it also lists the
\emph{authoritative records} that can settle the question, namely
records bound to the specific operation, such as the admission step's
own log, rather than a transcript that a correct run could equally have
produced, it can conclude that a violation is \emph{supported} when
every complete execution consistent with the available records violates
the rule, \emph{ruled out} when none does, and \emph{unknown} otherwise.
Removing records can then make a conclusion unknown but never reverse
it.  Finally, because this meaning is
fixed independently of any framework, only the binding between a
contract and a system's records changes across runtimes, so portability
becomes something to test rather than to assume.  We write contracts
for the four collective mechanisms that Han~\cite{han2026institutions}
identified as the places where collective reasoning breaks down:
\emph{access/routing}, whether required information reaches the agent
that needs it; \emph{admission/dependence}, which evidence is accepted
and how reports with a common origin are counted; \emph{state
maintenance}, whether accepted information is kept and updated; and
\emph{representation/action}, whether a correct decision is carried
out.  We call each of these a diagnostic \emph{locus}.

Testing a contract requires knowing whether the mechanism was in fact
violated, and natural runs rarely provide this ground truth.  Rerunning
a system with a broken mechanism does not help either: the language
models upstream would produce different content, and any change in the
outcome could come from that content rather than from the mechanism.
We therefore replay each execution up to the step where a mechanism
acts and continue it three ways from the identical prefix: unchanged,
with that one mechanism broken, and with it restored.  This supplies
ground truth by construction and lets us record separately whether the
mechanism was violated and whether the task succeeded.  To separate
what a diagnoser sees from how it reasons, we give every diagnoser the
same evidence packages and remove either irrelevant or decisive
records.  To test portability, we bind the unchanged contracts to
mechanisms implemented in independently developed multi-agent systems,
and accept a binding only when the system's own code implements the
rule and records the evidence the contract needs.  Throughout, we treat
a diagnosis as a claim that the execution record must license, not as a
label that the outcome suggests.

Our contributions are:
\begin{itemize}
  \item \textbf{Diagnostic contracts.}  We define when execution records
  support, rule out, or leave open a violation of a collective
  mechanism, and show that these conclusions can lose but never reverse
  their resolution when records are removed.
  \item \textbf{A controlled benchmark.}  We provide replayed executions
  with one mechanism broken or restored at a time across three task
  families and four agent frameworks, evidence packages that hold the
  records fixed, and bindings to independently developed systems.
  \item \textbf{Outcome is not mechanism; access is not use.}  In 43 of
  72 controlled admission violations the system still answered
  correctly, and a separately frozen allocation task showed the same
  pattern.  An LLM asked to diagnose without examples identified 26\% of
  broken mechanisms from public outputs and 63\% from internal records,
  yet with identical records a generic prompt often claimed conclusions
  the records did not support, whereas prompts stating the contracts
  largely avoided this on the outputs they completed.
  \item \textbf{Benchmark success need not transfer.}  The same frozen
  prompt kept its conclusion when irrelevant records were removed in all
  287 valid benchmark pairs but in only 13 of 28 pairs from an
  independently developed workflow, while still answering unknown whenever the decisive witness was deleted.  The settings differ in their
  tasks, so we read this as evidence of a transfer gap rather than an
  estimate of its size.
\end{itemize}
All results concern finite, controlled suites: the contracts cover four
mechanisms and exclude incentives, and they are tested one mechanism at
a time with supplied repairs.

\section{Related Work}
\label{sec:related}

\paragraph{Failure analysis in LLM-based multi-agent systems.}
The taxonomy, attribution, and localization methods discussed in
Section~\ref{sec:intro}~\cite{cemri2025mast,zhang2025whowhen,chen2026traceelephant,barke2026agentrx}
answer \emph{what kind} of failure occurred, \emph{who} caused it, and
\emph{when}.  We ask which records license the claim that a particular
collective mechanism was violated, including in successful runs; one
step can take part in several violations, and evidence that suffices for
attribution may leave that claim unresolved.

\paragraph{Fault injection and robustness.}
Fault-injection benchmarks measure how specified perturbations affect
detection, robustness, and recovery in multi-agent executions
\cite{jia2026masfire,zhang2026agentchaos,chen2026orchestra}, continuing
a longer line on fault tolerance and diagnosability~\cite{ezekiel2017fault}.
We also perturb executions, but we treat the injected condition as
experimental metadata rather than diagnostic evidence: a diagnosis must
be established from the records of the run, and the same perturbation
can leave the target predicate supported, ruled out, or unknown
depending on what those records expose.

\paragraph{Norm monitoring and model-based diagnosis.}
Normative multi-agent systems study when a monitor can detect norm
violations from the actions it observes.  Bulling et al.\ analyze
monitors and their relation to the norms they check
\cite{bulling2013monitoring}, and Alechina et al.\ synthesize optimal
approximations of norms for monitors with imperfect observational
capabilities~\cite{alechina2014approximation}.  Model-based and
plan-execution diagnosis likewise infer failed components from partial
observations~\cite{dejonge2009diagnosis,kalech2022modelbased}.  Our
contracts share the premise that observability bounds what can be
concluded, and our three-valued verdicts follow runtime verification of
partially observed traces~\cite{bauer2011runtime}.  We depart from this
tradition in three ways.  First, a contract's meaning is bound to
source-specific authoritative runtime records and their provenance,
rather than to observations of a known system model.  Second, the
benchmark uses shared-prefix perturbation and restoration to separate a
realized violation from both the task outcome and the injection
metadata.  Third, it measures empirically whether generic diagnosers
respect evidence obligations, namely uncertainty under
witness loss and invariance under irrelevant evidence, including on
independently developed workflows.  The formal properties in
Section~\ref{sec:object} ground these tests; they are not claimed as
new theory.

\paragraph{Collective mechanisms and runtime surfaces.}
The four loci come from the institutional account of
Han~\cite{han2026institutions}.  The present contribution is not the
taxonomy but its operationalization as executable, evidence-grounded
contracts.  AutoGen, LangGraph, CAMEL, and MetaGPT expose different
message, state, and orchestration records
\cite{wu2023autogen,langgraph2026,li2023camel,hong2023metagpt}; we use
them to test whether one diagnostic semantics can be witnessed through
heterogeneous runtime surfaces, not to assess the frameworks themselves.

\section{Diagnostic Contracts}
\label{sec:object}

A diagnostic contract answers two questions about one collective
mechanism: what counts as a violation, and which records an observer
needs in order to establish one.  Throughout, we
return to the admission example from the introduction.

\subsection{Mechanisms and violations}
\label{sec:loci}

Following Han~\cite{han2026institutions}, we describe the collective
mechanisms of a multi-agent system as an \emph{institution}
\begin{equation}
    I=\langle\rho,V,U,\Phi\rangle,
    \label{eq:institution}
\end{equation}
where $\rho$ decides which agent receives which information, $V$
decides which evidence is accepted and how reports with a common origin
are counted, $U$ decides how shared state is maintained and how agents are held accountable through incentives, and $\Phi$
decides how a decision becomes an action.  The four diagnostic loci are
$\loci=\{\rho,V,U_{\mathrm{state}},\Phi\}$, where $U_{\mathrm{state}}$
is the state-maintenance part of $U$; the incentive part of $U$ is
outside our scope.

Fix a task instance $q$.  A run of $q$ can unfold in many ways, and we
write $\Omega_q$ for the set of \emph{complete executions}, each a full
description of one run, including events that no log recorded.  For a
locus $\ell$, the \emph{violation predicate}
$p_{q\ell}:\Omega_q\to\{0,1\}$ equals $1$ exactly when execution $w$
breaks the rule of $\ell$.  In the admission example, $p_{q\ell}(w)=1$
when the admission step of $w$ counted the two wire copies as
independent.  A \emph{diagnostic contract} pairs this predicate with an
\emph{evidence specification}:
\begin{equation}
    \Gamma_{q\ell}=\bigl(p_{q\ell},\mathsf{Spec}_{q\ell}\bigr).
    \label{eq:contract}
\end{equation}
The specification states what an observer must know to evaluate the
predicate: when the contract applies, which task facts matter (here,
that the two copies share an origin), which record types are
authoritative (here, the admission log), and when those records are
complete (here, when the log covers every report in the decision).  The
predicate says what is true of a run; the specification says what must
be seen to know it.  A framework with a different log format may
therefore change whether a verdict can be reached, but not what the
verdict means.  Table~\ref{tab:loci} summarizes the four loci.

\begin{table}[t]
\caption{\textbf{The four diagnostic loci.}  For each locus: what counts
as a violation, which record can establish it, and a similar-looking
case that does not count.}
\label{tab:loci}
\centering
\footnotesize
\setlength{\tabcolsep}{3pt}
\begin{tabularx}{\columnwidth}{@{}>{\raggedright\arraybackslash}p{1.35cm}
  >{\raggedright\arraybackslash}X >{\raggedright\arraybackslash}X
  >{\raggedright\arraybackslash}X@{}}
\toprule
Locus & Violation & Authoritative record & Not a violation \\
\midrule
Access/ routing & Required information produced upstream is missing
  from the receiver's input & The receiver's complete input & Missing
  information the decision does not need \\
\addlinespace
Admission/ dependence & Evidence is accepted or weighted against the
  verification or common-origin rule & The admission decision and each
  report's origin & Suspicious reports that are correctly rejected or
  collapsed \\
\addlinespace
State maintenance & An accepted, relevant item is lost, stale, or
  overwritten & Shared state before and after the update & A change to
  state the task does not use \\
\addlinespace
Represen\-tation/ action & A correct decision is not carried out as a
  valid action & The submitted action and its receipt & An incorrect
  decision, to which the contract does not apply \\
\bottomrule
\end{tabularx}
\end{table}

Each predicate concerns one locus: a supported violation at one locus
says nothing about the others and does not by itself explain the final
outcome.  Several violations may coexist; our benchmark breaks one
mechanism at a time.  An agent's own reasoning errors fall outside
the contracts unless they also violate one of these rules.

\subsection{What counts as evidence}
\label{sec:evidence}

An observer never sees a complete execution, only what the system
recorded.  We represent what is available as an \emph{evidence bundle}
\begin{equation}
    E=(\mathcal R,\mathcal P,\mathcal C,\mathcal S).
    \label{eq:evidence}
\end{equation}
In the admission example, $\mathcal R$ holds the records, such as the
log entry stating that both copies were admitted; $\mathcal P$ holds
their provenance, namely that this entry was written by this run's
admission step about these two reports; $\mathcal C$ holds the
specifications, such as the admission rule and the fact that the copies
share an origin; and $\mathcal S$ states the scope of completeness, for
instance that the log lists every report the decision considered.

Two rules govern how a bundle may be read.  First, a missing record is
not a record of absence: if the receiver's input was not logged, the
bundle cannot show that an item was missing from it.  Second, authority
is specific to an operation.  A public transcript is not the input that
a particular agent received; a tool registry shows that a tool exists,
not that a caller could use it in this run; and a receipt counts only
when it is bound to the caller, payload, and operation it reports on.
Before any verdict is computed, an integrity check looks for records
that contradict one another or are bound to the wrong source.  A failed
check returns \texttt{measurement\_invalid}, which reports a broken
measurement rather than a verdict about the mechanism.

\subsection{From evidence to a verdict}
\label{sec:verdict}

Let $\mathcal W_{\Gamma}(E)\subseteq\Omega_q$ be the complete executions
that agree with everything in $E$, that is, the runs that could have
produced exactly these records.  The contract returns
\begin{equation}
d_{\Gamma}(E)=
\begin{cases}
    \texttt{supported},
        & p_{q\ell}(w)=1
          \text{ for every }w\in\mathcal W_{\Gamma}(E),\\
    \texttt{ruled\_out},
        & p_{q\ell}(w)=0
          \text{ for every }w\in\mathcal W_{\Gamma}(E),\\
    \texttt{unknown},
        & \text{otherwise}.
\end{cases}
\label{eq:semantics}
\end{equation}
In words, a violation is supported when every run consistent with the
records violates the rule, ruled out when none does, and unknown when
the records fit both kinds of run.  \texttt{unknown} is therefore a verdict in
its own right, correct whenever the records cannot tell the
possibilities apart.  In the admission
example, a complete admission log showing both copies admitted as
independent supports the violation; a log showing them collapsed rules
it out; and the public transcript alone, which both kinds of run could
have produced, leaves the verdict unknown.  The implementation computes
these verdicts with contract-specific checks rather than by enumerating
$\mathcal W_{\Gamma}(E)$.

\paragraph{Removing evidence.}
Diagnosers often see less than the full record, so we ask what happens
when records are removed.  A \emph{legal deletion} of $E$ removes some
records and keeps the rest unchanged: it is a bundle
$E'=(\mathcal R',\mathcal P',\mathcal C,\mathcal S')$ with
$\mathcal R'\subseteq\mathcal R$, provenance $\mathcal P'$ restricted to
$\mathcal R'$, and completeness claims $\mathcal S'$ only over scopes
already complete in $\mathcal S$ and still covered by $\mathcal R'$.
Because a missing record never counts as evidence of absence, removing
records can only admit more candidate runs:
\begin{equation}
    \mathcal W_{\Gamma}(E)\subseteq\mathcal W_{\Gamma}(E').
    \label{eq:info-order}
\end{equation}

\begin{proposition}[No opposite resolved flips]
\label{prop:monotone}
Let $E'$ be a legal deletion of $E$, and suppose that the contract
applies and the integrity check succeeds under both bundles.  If
$d_\Gamma(E')=\texttt{supported}$, then $d_\Gamma(E)=\texttt{supported}$;
if $d_\Gamma(E')=\texttt{ruled\_out}$, then
$d_\Gamma(E)=\texttt{ruled\_out}$.  Equivalently,
$d_\Gamma(E')\in\{d_\Gamma(E),\texttt{unknown}\}$.
\end{proposition}

\begin{proof}
Integrity under $E$ gives $\mathcal W_\Gamma(E)\neq\emptyset$, and
Equation~\eqref{eq:info-order} gives $\mathcal W_\Gamma(E)\subseteq
\mathcal W_\Gamma(E')$.  If $p_{q\ell}(w)=1$ for every
$w\in\mathcal W_\Gamma(E')$, the same holds on the nonempty subset
$\mathcal W_\Gamma(E)$; the ruled-out case is symmetric.
\end{proof}

\begin{corollary}[Irrelevant-evidence invariance]
\label{cor:irrelevant}
Suppose $d_\Gamma(E)$ is resolved and some legal deletion $F$ of $E$
already has $d_\Gamma(F)=d_\Gamma(E)$; call $F$ a \emph{witness}.  Then
every legal deletion $E'$ of $E$ that retains $F$, that is, whose
records and completeness claims include those of $F$, satisfies
$d_\Gamma(E')=d_\Gamma(E)$.
\end{corollary}

\begin{proof}
$F$ is a legal deletion of $E'$, so Proposition~\ref{prop:monotone}
applied to $(E',F)$ gives $d_\Gamma(E')=d_\Gamma(F)$.
\end{proof}

Both statements follow directly from Equation~\eqref{eq:semantics}.
Their role is to give two testable obligations to any diagnoser.  Removing
records may make a verdict unknown but must never reverse it; and
removing records outside a witness must not change the verdict at all.
Each resolved verdict is stored with the records and completeness
claims it relied on, so a reader can check both the claim and its
evidence.

\section{Benchmark Protocol and Challenge Suite}
\label{sec:protocol}

To test whether a diagnosis is correct, we must know whether the
mechanism was in fact violated.  We create this ground truth by
experiment: we take a run, freeze everything up to the step where one
mechanism acts, and continue the run three ways, once unchanged, once
with that mechanism broken, and once with it restored.  Panel~(c) of
Figure~\ref{fig:protocol} sketches the protocol.

\subsection{Tasks and held-out split}
\label{sec:suite}

Every task has the same small team: three source agents, each holding
some reports; an aggregator that combines the reports it receives into
a decision; and a submitter that turns the decision into an action.
Because each task has an exactly known answer, success is checked
without an LLM judge.  The \emph{canonical suite} contains two task
families.  In \emph{evidence aggregation}, the team chooses among four
verdicts from weighted reports; two independent verified reports support
the correct verdict, while two reports sharing one origin and one
unverified report support a wrong one, so the right answer requires the
admission rule to discard the unverified report and count the shared
pair once.  The execution in the introduction comes from this family.
In \emph{symbolic constraint}, the team identifies one of 24 candidates
from reports about its color, shape, and mark, again in the presence of
an unverified report and a pair of reports with a common origin.  A
third family, \emph{distributed resource allocation with commitments}
(DRAC), was frozen separately after the canonical evaluation: an
allocator assigns five tasks to agents under capacity limits, records
the assignment as a versioned reservation, and settles it through an
action; two independently written exact solvers agree on every DRAC
answer.

For the canonical suite, two families, four loci, and ten instances per
combination give 80 frozen challenges; a fixed rule reserves one
instance per combination for development and leaves 72 held out.  DRAC
has its own freeze of 40 challenges, 36 of them held out.  No held-out
challenge was run before the formal evaluation.  The split therefore
tests new instances of known families and mechanisms.

Each perturbation leaves the reports unchanged and breaks only the rule
of its target locus.  For access/routing, it removes one required
delivery to the aggregator.  For admission/dependence, it disables the
verification check or counts reports with a common origin separately.
For state maintenance, it skips an authoritative update or overwrites
an accepted item.  For representation/action, it keeps the aggregator's
decision but changes how the submitter encodes it, so that the action
is rejected.  If the aggregator's decision was already wrong, the
action contract does not apply, and the run is not counted as an
action failure.

\subsection{Replaying from a shared prefix}
\label{sec:prefix}

If we simply reran the system with a broken mechanism, the language
models upstream would produce different reports, and a change in the
outcome could come from those reports rather than from the mechanism.
We therefore replay one frozen \emph{prefix}, the part of the run before
the targeted step, under three versions of the institution.  For
challenge $q$, framework adapter $f$, and target locus $\ell$, let
$P_{qf}$ be that prefix, $T_{q\ell}$ the perturbation, and $R_{q\ell}$
the matched restoration:
\begin{equation}
    I_q^A = I_q,\qquad
    I_q^B = T_{q\ell}(I_q),\qquad
    I_q^C = R_{q\ell}\bigl(I_q^B\bigr).
    \label{eq:branches}
\end{equation}
Writing $\pi_k(I)$ for component $k$ of institution $I$, the design
requires that the three versions agree everywhere except at the target,
and that the restoration returns the target to its original value:
\begin{equation}
    \pi_k(I_q^A)=\pi_k(I_q^B)=\pi_k(I_q^C)\ \ \text{for } k\neq\ell,
    \qquad
    \pi_\ell(I_q^C)=\pi_\ell(I_q^A).
    \label{eq:coordinate-purity}
\end{equation}
Each branch $a\in\{A,B,C\}$ then continues from the same prefix,
$\tau^a_{qf}=\operatorname{Continue}_f(P_{qf},I_q^a;\xi_a)$, where
$\xi_a$ is the randomness of any model calls made after the prefix.
Branch comparisons are therefore conditioned on an identical history,
but later model calls may still differ; for representation/action, the
decision is part of the prefix, so no later model call is needed.
For the admission execution in the introduction, the prefix contains
the five reports as the source agents produced them, including an
unverified tip that the verification gate, intact in every branch,
excludes.  In branch $A$ the
admission step counts the two wire copies once, in $B$ it counts them
separately, and in $C$ the common-origin rule is restored.  The
reference verdict is ruled out, supported, and ruled out, while the
task succeeds in all three branches, because the aggregator collapsed
the copies itself in $B$.
Runtime checks confirm that the prefix is identical across branches,
that the perturbation took effect, and that the restoration returned
the target to its original value.

\subsection{What the frameworks contribute}
\label{sec:ownership}

In the canonical suite and DRAC, the benchmark implements the rules of
the institution, and each of the four adapters (AutoGen, LangGraph,
CAMEL, and MetaGPT) carries those rules through its framework's own
messages, state, and runtime records.  The comparison therefore asks
whether one diagnostic meaning can be read from four different kinds
of record, not how well the frameworks' default designs perform.  The
LangGraph adapter uses compiled replay, and the MetaGPT adapter covers
the core message runtime but not its optional tool stack.  The rules of independently
developed systems, which the benchmark does not implement, are tested
separately in Section~\ref{sec:res-binding}.

Labels that would reveal the answer are kept away from the diagnosers.
Agent inputs and diagnostic packages omit the branch identity, the
target locus, and the correct task answer, and diagnostic predictions
are frozen before they are matched with the true target.  The
reference implementation of the contracts reads the benchmark's own
mechanism records to check that the contracts behave as specified; it
serves as a conformance check rather than a competing diagnoser, and we
report each diagnoser together with the evidence it received.

\section{Evaluation Design}
\label{sec:evaluation}

The evaluation asks three questions, one for each finding.  (Q1) When a
contract supports a violation, does the task outcome reveal it?  (Q2)
Which records let a diagnoser detect a violation, and, given the same
records, does it respect the two obligations of
Proposition~\ref{prop:monotone} and Corollary~\ref{cor:irrelevant}?
(Q3) Do the contracts apply to mechanisms that others implemented, and
does a diagnoser's behavior carry over to their workflows?
Table~\ref{tab:glance} summarizes what each question varies, what it
holds fixed, and what it measures.  Intervals are 95\% percentile bootstrap intervals with 10,000
replicates that resample whole challenges (or, for independently developed systems,
whole tasks), so that
all branches and adapters of a challenge stay together;
Figure~\ref{fig:transfer} instead uses Wilson intervals over packaged
pairs, which ignore clustering by task and are therefore optimistic.

\begin{table}[t]
\caption{\textbf{The evaluation at a glance.}}
\label{tab:glance}
\centering
\footnotesize
\setlength{\tabcolsep}{3pt}
\begin{tabularx}{\columnwidth}{@{}>{\raggedright\arraybackslash}p{0.55cm}
  >{\raggedright\arraybackslash}X >{\raggedright\arraybackslash}X
  >{\raggedright\arraybackslash}X@{}}
\toprule
 & What changes & What is held fixed & What we measure \\
\midrule
Q1 & One mechanism: unchanged, broken, or restored & Everything before
  the targeted step & Reference verdict and task success on each
  branch \\
\addlinespace
Q2 & The records a diagnoser sees, and the diagnoser & The run being
  diagnosed & Detection rate; errors on missing or irrelevant
  records \\
\addlinespace
Q3 & The system: benchmark or independently developed & Contracts and
  diagnoser prompts & Accepted and rejected bindings; the same
  obligations as in Q2 \\
\bottomrule
\end{tabularx}
\end{table}

\subsection{Actor cohorts and outcome contrasts}

The canonical cohort runs each of the 72 held-out challenges on each of
the four framework adapters, giving 288 \emph{realizations}, each with
three branches.  All agents use \texttt{deepseek-flash}
(DeepSeek-V4.1-Flash lineage)~\cite{deepseek2026flash} with high
requested effort and a 1,000-token completion allowance.  The DRAC
cohort runs its 36 held-out challenges on the same adapters.  Because a
1,000-token allowance truncated nearly all allocator outputs, its formal
cohort uses a 16,384-token allowance chosen on development cases
without consulting correctness; we report the 137 of 144 realizations
whose source agents copied their reports exactly.  Reference verdicts
are computed afterwards from the recorded executions.

For branch $a\in\{A,B,C\}$, let $Y^a_{qf}\in\{0,1\}$ indicate whether
challenge $q$ on adapter $f$ produced exactly the correct answer.  We
measure how much breaking a mechanism hurts the outcome by the
control-averaged damage
\begin{equation}
  \Delta_{\mathrm{inst}}
  =\mathbb E_{q,f}\Bigl[\tfrac{1}{2}\bigl(Y^A_{qf}+Y^C_{qf}\bigr)-Y^B_{qf}\Bigr],
  \label{eq:originaleffects}
\end{equation}
where $\mathbb E_{q,f}$ averages uniformly over realizations.  For Q1,
the key count is the number of broken branches whose violation is
supported although the answer is still correct ($Y^B_{qf}=1$).

\subsection{Diagnosers and the records they see}

\paragraph{Evidence access.}
To measure how much the choice of records matters, an LLM diagnoser
without examples (prompt G0, which defines the four loci but not when
evidence is sufficient) reads each of the 288 broken canonical branches
in one of two views.  The \emph{public-output} view contains what the
agents said and the final action.  The \emph{native} view adds the
\emph{native records} that each framework keeps internally, such as each agent's
actual input, the admission decisions, the versions of shared state,
and the action receipts.  Both views use a JSON response format and a
65,536-token allowance calibrated on development cases.  Detection is
measured by target-locus recovery,
\begin{equation}
    \operatorname{Rec}_{m}(q,f)
    =\mathbf{1}\bigl\{\ell_q\in S_m(E^B_{qf})\bigr\},
    \label{eq:target-recovery}
\end{equation}
where $S_m(E)$ is the set of loci that diagnoser $m$ reports as
supported from view $E$ and $\ell_q$ is the locus that was broken.  To
test whether a diagnoser merely reacts to suspicious wording, we also
use 64 \emph{semantic near misses}, 16 per locus.  Each keeps a
surface cue of a violation while removing one condition it needs; for
example, a log may show a dropped message that the decision did not
need.

\paragraph{Same evidence.}
To separate what a diagnoser sees from how it reasons, we give several
diagnosers identical packages.  We build 384 base cases, 16 for each family, locus, and reference
verdict (supported or ruled out), chosen by a correctness-blind hash
rule, each in four versions drawn only from the
records the reference uses.  In the admission example, the full
version (F) contains the admission log, its provenance, and the claim
that the log is complete.  Version I removes records that do not bear
on the verdict, such as delivery records for other reports, while
keeping the admission log intact.  Versions M and S remove decisive
records, such as the log entry or its completeness claim, so that the
correct verdict becomes \texttt{unknown}.  Each version is obtained from F by deletion only,
with retained values and provenance checked for exact equality, so each
is a legal deletion.

Four diagnosers read the same packages: the deterministic reference
(R) and, on the same model, G0; G1, a prompt without examples that states each
contract and its three-valued semantics; and G2, which adds five
worked examples from development cases to G1.  An audit before any call
confirmed that R used no record unavailable to the others
(1,536/1,536).  We report four quantities: exact agreement with R;
\emph{unsupported certainty}, the share of cases where R says
\texttt{unknown} but the diagnoser still claims a violation is
supported or ruled out; \emph{opposite flips}, the number of times a
diagnoser's verdict switches between supported and ruled out as records
are removed, which Proposition~\ref{prop:monotone} forbids; and
\emph{irrelevance invariance}, the share of F$\to$I pairs whose verdict
does not change, which Corollary~\ref{cor:irrelevant} requires.  Under
preregistered allowances of 4,096 and 8,192 tokens, respectively
(65,536 for G0), G1 and G2 returned valid outputs
for 78.1\% and 67.1\% of packages, and missing outputs were
concentrated among packages whose correct verdict is
\texttt{supported}.  Both therefore fail the preregistered gate for
interpretability: their rates hold only for the outputs they
completed, and the study's preregistered verdict is
\emph{inconclusive}.

\subsection{Independently developed systems}

For each independently developed system, we read the pinned upstream source and fix
a candidate mapping from one of its mechanisms to a locus before
observing any outcome.  A mapping is accepted only if the system's own
code already implements the rule and records the evidence that the
unchanged contract requires; otherwise it is rejected or narrowed.  We
call an accepted mapping an \emph{external binding}.
Two independent, outcome-blind LLM reviews (prompts in the supplement)
assess each candidate, and when they judge a mapping plausible but ambiguous, acceptance also requires
development evidence that resolves the stated concern before any
held-out run.  Our adapters for external bindings only observe, record,
inject the fault, restore, and check the task; for the two fully
controlled systems, a line-by-line audit finds no institutional logic
in them.  Each external binding is tested on held-out A/B/C triplets
forked from a native checkpoint.  For two multi-agent workflows, we
also build full (F), irrelevant-deletion (I), and witness-deletion (W)
packages, obtain correct verdicts from an independent adjudicator that
applies the system's own rule, and run the frozen G1 and G2 prompts
unchanged.

\section{Results}
\label{sec:results}

\subsection{Task outcome is not mechanism correctness}
\label{sec:res-outcome}

\begin{table}[t]
\caption{\textbf{Mechanism verdicts and task outcomes.}  B support:
target violation supported on broken branches; control activation:
target supported on unchanged or restored branches; correct under
violation: supported violations with a correct task outcome.
Canonical branches use a 1,000-token completion allowance
(Section~\ref{sec:res-outcome}).}
\label{tab:formal}
\centering
\small
\setlength{\tabcolsep}{3.5pt}
\begin{tabular}{@{}lcccc@{}}
\toprule
 & B support & \shortstack{Control\\activation} &
 \shortstack{Correct under\\violation} & $\Delta_{\mathrm{inst}}$ \\
\midrule
Routing   & 72/72   & 0/144 & 0/72   & 1.000 \\
Admission & 72/72   & 0/144 & 43/72  & 0.403 \\
State     & 72/72   & 0/144 & 0/72   & 1.000 \\
Action    & 72/72   & 0/144 & 0/72   & 1.000 \\
Canonical & 288/288 & 0/576 & 43/288 & 0.851 \\
\midrule
DRAC      & 137/137 & 0/274 & 28/137 & 0.796 \\
\bottomrule
\end{tabular}
\end{table}

All 288 canonical realizations passed the runtime checks, and the
reference supported the target violation on every broken branch and
on none of the 576 unchanged and restored controls (Table~\ref{tab:formal}).
Because the reference reads the records that the harness writes, these
counts verify implementation conformance, not diagnostic skill.
Removing a required delivery, overriding the authoritative state, or
rejecting the terminal action leaves the system no route to the
correct answer, so these three loci serve as positive controls: every
such broken branch failed.

Admission violations, in contrast, often leave the answer correct.  In
43 of 72 admission perturbations the contract supported the violation
and the task still succeeded, so the outcome concealed 59.7\% of these
violations.  DRAC, whose allocator optimizes a capacitated assignment
and commits it as a versioned reservation, shows the same pattern: 28
supported violations still produced the exact optimal allocation.
Concealment persists under a larger completion allowance.  Of the 29
failed admission branches, 24 ended without an answer at the
1,000-token limit, so the admission damage in Table~\ref{tab:formal}
partly reflects that limit.  Replaying the 288 broken branches from the same prefixes
with a 65,536-token allowance removed all truncation: 58 of 72
admission branches then answered correctly, every restored branch
succeeded, and every target violation remained supported.  Raising the
requested effort at that allowance moved success only to 61 of 288
(details in the supplement).  Outside the canonical suite, every
broken branch in the fully controlled MetaGPT and Agents SDK
validations of Section~\ref{sec:res-binding} failed its native task
oracle, and the native-workflow transfer study did not adjudicate task
correctness.

\subsection{Native records enable detection}
\label{sec:diagnosisresults}

\begin{figure}[t]
\centering
\definecolor{vizBlue}{HTML}{2A78D6}
\definecolor{vizOrange}{HTML}{EB6834}
\definecolor{vizInk}{HTML}{52514E}
\definecolor{vizGrid}{HTML}{E4E3DF}
\begin{tikzpicture}
\begin{axis}[
  width=\columnwidth, height=3.9cm, scale only axis=false,
  xmin=-0.16, xmax=1.18, ymin=-0.6, ymax=3.6,
  xtick={0,0.25,0.5,0.75,1}, xticklabels={0,.25,.50,.75,1},
  ytick={0,1,2,3}, yticklabels={Action,State,Admission,Routing},
  xlabel={Target-locus recovery on broken branches},
  axis line style={vizInk}, tick style={vizInk, line width=0.4pt},
  ticklabel style={font=\footnotesize, text=vizInk}, label style={font=\footnotesize, text=vizInk},
  xmajorgrids, grid style={vizGrid, line width=0.4pt}, axis x line*=bottom, axis y line*=left,
  legend style={font=\footnotesize, draw=none, fill=none, at={(0.5,1.02)}, anchor=south, legend columns=2,
                /tikz/every even column/.append style={column sep=10pt}},
  clip=false,
]
  \draw[vizGrid, line width=2.2pt] (axis cs:1.0000,0) -- (axis cs:1.0000,0);
  \draw[vizGrid, line width=2.2pt] (axis cs:0.0000,1) -- (axis cs:0.1528,1);
  \draw[vizGrid, line width=2.2pt] (axis cs:0.0556,2) -- (axis cs:0.3750,2);
  \draw[vizGrid, line width=2.2pt] (axis cs:0.0000,3) -- (axis cs:1.0000,3);
\addplot[only marks, mark=*, mark size=2.4pt, color=vizOrange, mark options={draw=white, line width=0.6pt},
         error bars/.cd, x dir=both, x explicit, error bar style={vizOrange, line width=0.6pt}]
  coordinates {(1.0000,0) -= (0.0000,0) += (0.0000,0) (0.0000,1) -= (0.0000,0) += (0.0000,0) (0.0556,2) -= (0.0556,0) += (0.0972,0) (0.0000,3) -= (0.0000,0) += (0.0000,0)};
\addlegendentry{Public outputs only}
\addplot[only marks, mark=*, mark size=2.4pt, color=vizBlue, mark options={draw=white, line width=0.6pt},
         error bars/.cd, x dir=both, x explicit, error bar style={vizBlue, line width=0.6pt}]
  coordinates {(1.0000,0) -= (0.0000,0) += (0.0000,0) (0.1528,1) -= (0.0972,0) += (0.1250,0) (0.3750,2) -= (0.1389,0) += (0.1528,0) (1.0000,3) -= (0.0000,0) += (0.0000,0)};
\addlegendentry{Native records}
  \node[font=\scriptsize, text=vizInk, anchor=west] at (axis cs:1.0000,0) {\,72/72 both};
  \node[font=\scriptsize, text=vizInk, anchor=east] at (axis cs:0.0000,1) {0/72\,};
  \node[font=\scriptsize, text=vizInk, anchor=west] at (axis cs:0.2778,1) {\,11/72};
  \node[font=\scriptsize, text=vizInk, anchor=east] at (axis cs:0.0000,2) {4/72\,};
  \node[font=\scriptsize, text=vizInk, anchor=west] at (axis cs:0.5278,2) {\,27/72};
  \node[font=\scriptsize, text=vizInk, anchor=east] at (axis cs:0.0000,3) {0/72\,};
  \node[font=\scriptsize, text=vizInk, anchor=west] at (axis cs:1.0000,3) {\,72/72};
\end{axis}
\end{tikzpicture}
\caption{\textbf{Evidence access by locus.}  Target-locus recovery of
the G0 diagnoser on the 288 broken canonical branches (72
per locus) from public outputs and from native records.
Whiskers: 95\% bootstrap intervals over the 18 challenges of each
locus.}
\Description{Dot plot with one row per locus.  Public outputs recover
0/72 routing, 4/72 admission, 0/72 state, and 72/72 action failures;
native records recover 72/72, 27/72, 11/72, and 72/72.}
\label{fig:access}
\end{figure}

Native records more than doubled what G0 could detect.  With
public outputs alone, G0 recovered 76 of 288 broken mechanisms
(26.4\%); with native records it recovered 182 (63.2\%), a paired
increase of $0.368\,[0.271,0.472]$.  Figure~\ref{fig:access} locates
the gain.  Public outputs exposed almost only action failures, whose
rejection is visible; native records made routing fully recoverable but
left admission and state maintenance mostly unrecovered.  Surface cues
do not explain the gain: on the 64 semantic near misses, which keep a
cue but remove a condition the violation needs, the two views activated
the target in only 2 and 1 cases.

\subsection{Generic diagnosers assert verdicts the records do not license}
\label{sec:res-same}

\begin{table}[t]
\caption{\textbf{Same evidence, different diagnosers.}  Rates for G1
and G2 are conditional on valid outputs; both fail the preregistered
interpretability gate (Section~\ref{sec:evaluation}).  Unsupp.\
cert.: resolved verdict where R returns \texttt{unknown}.  Flips:
opposite resolved verdicts along deletion edges.}
\label{tab:same-evidence}
\centering
\small
\setlength{\tabcolsep}{3pt}
\begin{tabular}{@{}lccccc@{}}
\toprule
 & Valid & \shortstack{Agree\\with R} & \shortstack{Unsupp.\\cert.} &
 Flips & \shortstack{F$\to$I\\invariant} \\
\midrule
G0 & 1536/1536 & 764/1536  & 576/768 & 138/1152 & 259/384 \\
G1 & 1199/1536 & 1139/1199 & 53/533  & 0/748    & 287/287 \\
G2 & 1031/1536 & 1022/1031 & 8/435   & 0/561    & 227/228 \\
\bottomrule
\end{tabular}
\end{table}

Access alone does not determine the diagnosis
(Table~\ref{tab:same-evidence}).  Given exactly the records that the
contract uses, G0 recovered all 192 supported violations and activated
on none of the 192 ruled-out cases, yet it broke both obligations.
When decisive witnesses were deleted, it asserted a resolved verdict on
576 of 768 packages whose evidence forces none and reversed its
resolved verdict along 138 of 1,152 deletion edges, against
Proposition~\ref{prop:monotone}; when only records outside a retained
witness were removed, it changed 125 of 384 verdicts, against
Corollary~\ref{cor:irrelevant}.  With the contracts stated in the
prompt, these errors largely disappeared: on their valid outputs, G1
and G2 made no opposite flips, preserved their verdicts under
irrelevant deletion in all but one case, and resolved far fewer unknown
cases (53/533 and 8/435).  Identical evidence therefore does not fix
the diagnosis; the diagnoser that reads it determines whether the
contract obligations hold.

\subsection{Contracts bind to independently developed systems}
\label{sec:res-binding}

\begin{table}[t]
\caption{\textbf{External bindings.}  Held-out A/B/C triplets forked
from native checkpoints; entries give target support on B and
activation on A and C.  $^\dagger$Both LLM reviews judged this
mapping plausible but ambiguous; it was accepted on development
evidence.  Co-STORM yielded 18 qualified triplets from 24 held-out tasks;
the table reports the 9 of them that fall in the 12-task diagnostic
subset fixed by a hash rule before held-out execution, whose other three
tasks had no qualified triplet (one lacked an original pre-operation
snapshot and two failed native qualification).  Rejected candidates are
listed below the rule.}
\label{tab:external}
\centering
\small
\setlength{\tabcolsep}{3pt}
\begin{tabular}{@{}llcc@{}}
\toprule
Native mechanism & Locus & $n$ & B / A / C \\
\midrule
MetaGPT recipient-local delivery & routing   & 8  & 8 / 0 / 0 \\
MetaGPT memory replay suppression & state    & 8  & 8 / 0 / 0 \\
Agents SDK tool call$\to$receipt  & action    & 12 & 12 / 0 / 0 \\
Co-STORM citation membership$^\dagger$ & admission & 9  & 9 / 0 / 0 \\
\midrule
Agents SDK handoff / session      & routing / state & \multicolumn{2}{c}{rejected} \\
Agents SDK approval gate          & admission & \multicolumn{2}{c}{rejected} \\
Workflow without admission guard  & admission & \multicolumn{2}{c}{rejected} \\
\bottomrule
\end{tabular}
\end{table}

\paragraph{Accepted external bindings.}
The unchanged contracts bound to four native mechanisms in three
independently developed systems: MetaGPT~\cite{hong2023metagpt}, the
OpenAI Agents SDK~\cite{openai2026agentssdk}, and
Co-STORM~\cite{jiang2024costorm} (Table~\ref{tab:external}).  In each,
the upstream code already implemented a contract's rule and exposed its
authoritative records, and together the bindings cover all four loci,
each in a narrow form.  On held-out triplets forked from native
checkpoints, the contracts supported every broken branch and no
unchanged or restored control, and withholding the native tool receipt
returned \texttt{unknown} in all 18 action projections (A/B/C of six
triplets).  Both LLM reviews judged three of the four mappings natural.
The admission mapping is the least secure: both reviews judged it
plausible but ambiguous, and it was accepted only after development
evidence resolved their concern.  The action binding exercises the
SDK's own function-tool lifecycle with a deterministic test fixture
rather than inside a multi-agent workflow.

\paragraph{Rejected bindings.}
The acceptance rule also rejected candidates whose native mechanism did
not implement the contract's rule: generic handoff and session
persistence would have needed policies written by the experimenters, an
approval gate authorizes actions rather than admitting evidence, and
one audited workflow exposed no separable admission guard.  A
post-execution source audit narrowed a further workflow's measured
predicate from admission to a successor-routing subtype, so we report
that workflow only as exploratory evidence in the supplement.  Released
trace collections rarely expose the records that the contracts require:
none of 156 TraceElephant MAS traces supports any locus; 147 remain
unresolved, and the other 9 only rule out a narrow routing
failure~\cite{chen2026traceelephant}.

\subsection{One obligation transfers, the other does not}
\label{sec:res-transfer}

\begin{figure}[t]
\centering
\definecolor{vizBlue}{HTML}{2A78D6}
\definecolor{vizOrange}{HTML}{EB6834}
\definecolor{vizInk}{HTML}{52514E}
\definecolor{vizGrid}{HTML}{E4E3DF}
\begin{tikzpicture}
\begin{groupplot}[
  group style={group size=2 by 1, horizontal sep=0.9cm},
  width=0.54\columnwidth, height=3.9cm,
  xmin=-0.6, xmax=1.6, ymin=0, ymax=1.2,
  xtick={0,1}, xticklabels={Canonical, Co-STORM},
  ytick={0,0.5,1},
  axis line style={vizInk}, tick style={vizInk, line width=0.4pt},
  ticklabel style={font=\footnotesize, text=vizInk},
  title style={font=\footnotesize, text=vizInk, yshift=-2pt},
  ymajorgrids, grid style={vizGrid, line width=0.4pt}, axis x line*=bottom, axis y line*=left,
  clip=false,
]
\nextgroupplot[title={Witness deleted $\to$ \texttt{unknown}}, yticklabels={0,.50,1}]
\addplot[color=vizBlue, line width=1pt, mark=*, mark size=2.2pt, mark options={draw=white, line width=0.5pt}, error bars/.cd, y dir=both, y explicit, error bar style={vizBlue, line width=0.6pt}] coordinates {(-0.07,1.0000) -= (0,0.0135) += (0,0.0000) (0.9299999999999999,1.0000) -= (0,0.1206) += (0,0.0000)};
\node[font=\scriptsize, text=vizInk, anchor=south, yshift=2.5pt] at (axis cs:-0.07,1.0000) {280/280};
\node[font=\scriptsize, text=vizInk, anchor=south, yshift=2.5pt] at (axis cs:0.9299999999999999,1.0000) {28/28};
\addplot[color=vizOrange, line width=1pt, mark=square*, mark size=2.2pt, mark options={draw=white, line width=0.5pt}, error bars/.cd, y dir=both, y explicit, error bar style={vizOrange, line width=0.6pt}] coordinates {(0.07,1.0000) -= (0,0.0166) += (0,0.0000) (1.07,1.0000) -= (0,0.1487) += (0,0.0000)};
\node[font=\scriptsize, text=vizInk, anchor=north, yshift=-2.5pt] at (axis cs:0.07,1.0000) {228/228};
\node[font=\scriptsize, text=vizInk, anchor=north, yshift=-2.5pt] at (axis cs:1.07,1.0000) {22/22};
\nextgroupplot[title={Irrelevant deleted $\to$ verdict kept}, yticklabels={0,.50,1}]
\addplot[color=vizBlue, line width=1pt, mark=*, mark size=2.2pt, mark options={draw=white, line width=0.5pt}, error bars/.cd, y dir=both, y explicit, error bar style={vizBlue, line width=0.6pt}] coordinates {(-0.07,1.0000) -= (0,0.0132) += (0,0.0000) (0.9299999999999999,0.4643) -= (0,0.1690) += (0,0.1776)};
\node[font=\scriptsize, text=vizInk, anchor=south, yshift=2.5pt] at (axis cs:-0.07,1.0000) {287/287};
\node[font=\scriptsize, text=vizInk, anchor=west, xshift=4pt] at (axis cs:0.9299999999999999,0.4643) {13/28};
\addplot[color=vizOrange, line width=1pt, mark=square*, mark size=2.2pt, mark options={draw=white, line width=0.5pt}, error bars/.cd, y dir=both, y explicit, error bar style={vizOrange, line width=0.6pt}] coordinates {(0.07,0.9956) -= (0,0.0200) += (0,0.0036) (1.07,0.3810) -= (0,0.1734) += (0,0.2103)};
\node[font=\scriptsize, text=vizInk, anchor=north, yshift=-2.5pt] at (axis cs:0.07,0.9956) {227/228};
\node[font=\scriptsize, text=vizInk, anchor=west, xshift=4pt] at (axis cs:1.07,0.3810) {8/21};
\end{groupplot}
\node[font=\footnotesize, text=vizInk, anchor=north] at ($(group c1r1.south)!0.5!(group c2r1.south)+(0,-0.55cm)$)
  {\tikz[baseline=-0.6ex]{\draw[vizBlue, line width=1pt] (0,0) -- (0.4,0);
     \fill[vizBlue] (0.2,0) circle (2pt);}\,G1 (zero-shot, three-valued)\qquad
   \tikz[baseline=-0.6ex]{\draw[vizOrange, line width=1pt] (0,0) -- (0.4,0);
     \fill[vizOrange] (0.13,-0.07) rectangle (0.27,0.07);}\,G2 (G1 + demonstrations)};
\end{tikzpicture}
\caption{\textbf{The same frozen diagnosers on canonical and
independently developed workflow evidence.}  Left: share of valid pairs in which deleting the
decisive witness yields \texttt{unknown} (canonical: F resolved, S
\texttt{unknown}).  Right: share in which deleting records outside a
retained witness leaves the verdict unchanged.  Whiskers: Wilson 95\%
intervals.  Canonical rates are conditional on valid outputs.  Of the
28 Co-STORM pairs, 27 come from the nine triplets in
Table~\ref{tab:external} and one from a native record without a
pre-operation snapshot.}
\Description{Two slope panels comparing canonical and Co-STORM rates.
Witness deletion stays at 1.0 for G1 (280/280 to 28/28) and G2 (228/228
to 22/22).  Irrelevant-evidence invariance falls for G1 from 287/287 to
13/28 and for G2 from 227/228 to 8/21.}
\label{fig:transfer}
\end{figure}

Of the two obligations that G1 satisfied on the canonical packages, one
held on evidence from an independently developed workflow and one did
not (Figure~\ref{fig:transfer}).  Witness deletion transferred: on the
Co-STORM packages, R, G1, and G2 returned \texttt{unknown} on every
jointly valid witness-deletion pair (28/28, 28/28, and 22/22).
Irrelevant-evidence invariance failed.  The same frozen G1 diagnoser,
which preserved its verdict on all 287 valid canonical F$\to$I pairs,
preserved it on only 13 of 28 Co-STORM pairs (Wilson 95\% interval
0.30--0.64), whereas R preserved all 28.  G2 preserved 8 of
21 and, on full evidence, matched the adjudicated verdict less often
than G1 (net $-4$ over 25 jointly valid packages).  Because the
canonical and native tasks differ, these contrasts reveal a transfer gap
but do not estimate its size.

\section{Discussion and Limitations}
\label{sec:discussion}

\paragraph{Canonical construction.}
The benchmark authors the canonical institutions, the role topology is
compact, and instances within a family come from one generator.  Three
loci therefore act as positive controls, and the coexistence of
violation and correct outcome is an existence result; we do not
estimate how often it occurs in deployed systems.

\paragraph{Diagnostic evidence.}
All LLM diagnosers use one model lineage.  G1 and G2 rates are
conditional on incomplete outputs, and the preregistered verdict is
inconclusive.  The transfer result rests on one independently developed
workflow with 28 packaged pairs from 10 tasks, whose admission binding
both LLM reviews found ambiguous; a second workflow is exploratory
because its locus was narrowed after execution.

\paragraph{Scope.}
The contracts cover four loci and exclude incentives.  Interventions
perturb one mechanism at a time, restorations are supplied rather than
discovered, and every external binding required a manual source audit.
Discovering bindings automatically, diagnosing simultaneous violations,
and extending the contracts to open-ended tasks without exact oracles
remain open problems.

\section{Conclusion}
\label{sec:conclusion}

The execution that opened this paper answered correctly while its
admission rule was violated, and only the admission log could say so.
Diagnostic contracts turn that observation into a method: they state
what a violation is, which records can establish it, and when the honest
verdict is \emph{unknown}.  A correct answer does not certify the
mechanisms behind it; a diagnosis is licensed by authoritative records,
not by the outcome or a label; and diagnostic behavior that looks solved
on a benchmark has to be retested on the system where it will be used.


\bibliographystyle{ACM-Reference-Format}
\bibliography{references}

\end{document}